\documentclass[english]{llncs}
\usepackage{pgfplots}
\usepackage{amssymb}
\usepackage{amsmath}
\usepackage{colortbl}
\usepackage{xspace}

\usepackage[ruled, linesnumbered]{algorithm2e}

\newcommand{\shorten}[1]{\textcolor[rgb]{0.65,0.65,0.65}{#1}\xspace}

\definecolor{firstcolor}{RGB}{255,255,204}
\definecolor{secondcolor}{RGB}{161,218,180}
\definecolor{thirdcolor}{RGB}{65,182,196}
\definecolor{fourthcolor}{RGB}{34,94,168}

\renewcommand{\shorten}[1]{}

\newcommand{\argmax}{\operatorname{argmax}}
\newcommand{\argmin}{\operatorname{argmin}}

\newcommand{\ie}{i.\,e.\xspace}

\newcommand{\etal}{et al.\xspace}

\newcommand{\bigO}{\mathcal{O}}

\author{Moritz von Looz\inst{1} \and Mario Wolter\inst{2} \and Christoph R. Jacob\inst{2} \and Henning Meyerhenke\inst{1}}
\title{Better partitions of protein graphs for subsystem quantum chemistry}

\institute{\email{\{moritz.looz-corswarem, meyerhenke\}@kit.edu} \\ Institute of Theoretical Informatics\\ Karlsruhe Institute of Technology (KIT), Germany \and \email{\{m.wolter, c.jacob\}@tu-braunschweig.de} \\ Institute of Physical and Theoretical Chemistry\\ TU Braunschweig, Germany}

\begin{document}
 \maketitle
 \begin{abstract}
  Determining the interaction strength between proteins and small molecules is key to analyzing their biological function.
  Quantum-mechanical calculations such as \emph{Density Functional Theory} (DFT) give accurate and theoretically well-founded results. 
With common implementations the running time of DFT calculations increases quadratically with molecule size.
  Thus, numerous subsystem-based approaches have been developed to accelerate quantum-chemical calculations.
  These approaches partition the protein into different fragments, which are treated separately.
Interactions between different fragments are approximated and introduce inaccuracies in the calculated interaction energies.

  To minimize these inaccuracies, we represent the amino acids and their interactions as a weighted graph in order to apply graph partitioning.
  None of the existing graph partitioning work can be directly used, though, due to the unique constraints in partitioning such protein graphs.
We therefore present and evaluate several algorithms, partially building upon established concepts, but adapted to handle the new constraints.
  For the special case of partitioning a protein along the main chain, we also present an efficient dynamic programming algorithm 
that yields provably optimal results. In the general scenario our algorithms usually improve the previous approach significantly and take at most a few seconds.
 \end{abstract}

 \section{Introduction}
 \label{sec:introduction}
 \paragraph{Context.}
  The biological role of proteins is largely determined by their interactions with other proteins and small molecules. 
Quantum-chemical methods, such as \emph{Density Functional Theory} (DFT), provide an accurate description of these interactions based on quantum mechanics. 
A major drawback of DFT is its time complexity, which has been shown to be cubic with respect to the protein size in the worst case~\cite{cramer-book,jensen-book}.
For special cases this complexity can be reduced to being linear~\cite{ochsenfeld_linear-scaling_2007,adf-lin-scaling}. 
DFT implementations used for calculations on proteins are in between these bounds and typically show quadratic behavior with significant constant factors, rendering proteins bigger than a few hundred amino acids prohibitively expensive to compute~\cite{cramer-book,jensen-book}. 
As an example, Figure~\ref{plot:time-simulation} in Appendix~\ref{appx:run-time-chem} shows an excerpt from experimental running times 
of quantum-chemical calculations on protein fragments which support this quadratic dependence. 

To mitigate the computational cost, quantum-chemical subsystem methods have been developed~\cite{gordon_fragmentation_2012,jacob_subsystem_2014}.
In such approaches, large molecules are separated into fragments (= subsystems) which are then treated individually.
A common way to deal with individual fragments is to assume that they do not interact with each other.
The error this introduces for protein--protein or protein--molecule interaction energies (or for other local molecular properties of interest) depends on the size and location of fragments: A partition that cuts right through the strongest interaction in a molecule will give worse results than one that carefully avoids this.
It should also be considered that a protein consists of a \emph{main chain} (also called \emph{backbone}) of amino acids. This main chain folds into 3D-secondary-structures,
stabilized by non-bonding interactions (those not on the backbone) between the individual amino acids. These different connection types (backbone vs non-backbone) have different influence on the interaction energies.

\paragraph{Motivation.}
Subsystem methods are very powerful in quantum chemistry~\cite{gordon_fragmentation_2012,jacob_subsystem_2014} but so far require manual cuts with chemical insight to achieve good partitions~\cite{karin-3fde}. 
Currently, when automating the process, domain scientists typically cut every $X$ amino acids along the main chain (which we will call the \emph{naive approach} in the following). This gives in general suboptimal and unpredictable results (see Figure~\ref{plot:frag_err_naive}
in Appendix~\ref{sec:illustrate}). 

By considering amino acids as nodes connected by edges weighted with the expected error in the interaction energies, one can construct (dense) graphs representing the proteins. Graph partitions with a light cut, \ie partitions of the vertex set whose inter-fragment edges have low total weight, should then 
correspond to a low error for interaction energies.
A general solution to this problem has high significance, since it is applicable to any subsystem-based method and since it will enable such calculations on larger systems with controlled accuracy.
Yet, while several established graph partitioning algorithms exist, none of them is directly applicable to our problem scenarios due to additional domain-specific optimization constraints
(which are outlined in Section~\ref{sec:problem}).

\newcommand{\maxSize}{\ensuremath{\mathrm{maxSize}}}

 \paragraph{Contributions.}
For the first of two problem scenarios, the special case of continuous fragments along the main chain, we provide in Section~\ref{sec:algo-main-chain}
a dynamic programming (DP) algorithm. We prove that it yields an optimal solution with a worst-case time complexity of $\bigO(n^2 \cdot \maxSize)$.

For the general protein partitioning problem, we provide three algorithms using established partitioning concepts, now equipped with techniques
for adhering to the new constraints (see Section~\ref{sec:algo-general}):
(i) a greedy agglomerative method,
(ii) a multilevel algorithm with Fiduccia-Mattheyses~\cite{FM82} refinement, and
(iii) a simple postprocessing step that ``repairs'' traditional graph partitions.

Our experiments (Section~\ref{sec:experiments}) use several protein graphs representative for DFT calculations. Their number of nodes
is rather small (up to 357), but they are complete graphs. The results show that our algorithms are usually better in quality than the naive 
approach. While none of the new algorithms is consistently the best one, the DP algorithm can be called most robust since it is always 
better in quality than the naive approach.
A meta algorithm that runs all single algorithms and picks the best solution would still take only about ten seconds per instance and 
improve the naive approach on average by  $13.5\%$ to $20\%$, depending on the imbalance.
In the whole quantum-chemical workflow the total partitioning time of this meta algorithm is still small.
Parts of this paper have been published in the proceedings of the 15th International Symposium on Experimental Algorithms (SEA 2016).

 \section{Problem Description}
\label{sec:problem}
\renewcommand{\P}{$\mathcal{P}$}
\newcommand{\NP}{$\mathcal{NP}$}
Given an undirected connected graph $G=(V,E)$ with $n$ nodes and $m$ edges, a set of $k$ disjoint non-empty node subsets $V_1, V_2, ... V_k$ is called a $k$-partition of $G$ if the union of the subsets yields $V$ ($V = \bigcup_{1\leq i \leq k} V_i$). We denote partitions with the letter $\Pi$ and call the subsets \emph{fragments} in this paper.

Let $w(u,v)$ be the weight of edge $\{u,v\} \in E$, or 1 in an unweighted graph.
Then, the \emph{cut weight} of a graph partition is the sum of the weights of edges with endpoints in different subsets: 
$\mathrm{cutweight}(\Pi, G) = \sum_{u \in V_i, v \in V_j, i \not= j, V_i, V_j \in \Pi}w(u,v)$.
The largest fragment's size should not exceed $\maxSize := (1+\epsilon) \cdot \lceil n / k \rceil$, where $\epsilon$ is the so-called
\emph{imbalance} parameter. A partition is balanced iff $\epsilon = 0$.

Given a graph $G=(V,E)$ and $k \in \mathbb{N}_{\geq 2}$, \emph{graph partitioning} is often defined as the problem of finding a 
$k$-partition with minimum cut weight while respecting the constraint of maximum imbalance $\epsilon$.
This problem is \NP-hard~\cite{GareyJS74some} for general graphs and values of $\epsilon$.
For the case of $\epsilon = 0$, no polynomial time algorithm can deliver a constant factor approximation guarantee unless \P{} equals \NP{}~\cite{Andreev2006}.

\subsection{Protein Partitioning}
We represent a protein as a weighted undirected graph. Nodes represent amino acids, edges represent bonds or other interactions.
(Note that our graphs are different from protein interaction networks~\cite{Pavlopoulos2011}.)
Edge weights are determined both by the strength of the bond or interaction and the importance of this edge to the protein function.
Such a graph can be constructed from the geometrical structure of the protein using chemical heuristics whose detailed discussion is
beyond our scope.
Partitioning into fragments yields faster running time for DFT since the time required for a fragment is quadratic in its size.
The cut weight of a partition corresponds to the total error caused by dividing this protein into fragments.
A balanced partition is desirable as it maximizes this acceleration effect. However, relaxing the constraint with a small $\epsilon > 0$ makes
sense as this usually helps in obtaining solutions with a lower error.

Note that the positions on the main chain define an ordering of the nodes. From now on we assume the nodes to be 
numbered along the chain.

\paragraph{New Constraints.}
Established graph partitioning tools using the model of the previous section cannot be applied directly to our problem since
protein partitioning introduces additional constraints and an incompatible scenario due to chemical idiosyncrasies:
\begin{itemize}
\item The first constraint is caused by so-called \emph{cap molecules} added for the subsystem calculation. These cap molecules are added at fragment boundaries 
(only in the DFT, not in our graph) to obtain chemically meaningful fragments. This means for the graph that if node $i$ and node $i+2$ belong
to the same fragment, node $i+1$ must also belong to that fragment. Otherwise the introduced cap molecules will overlap spatially and therefore not represent a chemically meaningful structure. We call this the \emph{gap} constraint.
\item More importantly, some graph nodes can have a charge. It is difficult to obtain robust convergence in quantum-mechanical calculations for fragments with more than one charge. Therefore, together with the graph a (possibly empty) list of charged nodes is given and two charged nodes must not be in the same fragment. This is called the \emph{charge} constraint.
\end{itemize}

We consider here \textbf{two problem scenarios} (with different chemical interpretations) in the context of protein partitioning:
\begin{itemize}
\item \textbf{Partitioning along the main chain:} 
The main chain of a protein gives a natural structure to it.
We thus consider a scenario where partition fragments are forced to be continuous on the main chain.
This minimizes the number of cap molecules necessary for the simulation and has the additional advantage of better comparability 
with the naive partition.

Formally, the problem can be stated like this:
Given a graph $G=(V,E)$ with ascending node IDs according to the node's main chain position, an integer $k$ and a maximum imbalance $\epsilon$,
find a $k$-partition with minimum cut weight such that $v_j \in V_i \wedge v_j+l \in V_i \rightarrow v_j+1 \in V_i, 1 \leq j \leq n, l \in \mathbb{N}^+, 1 \leq i \leq k$
and which respects the balance, gap, and charge constraints.

\item \textbf{General protein partitioning:} 
The general problem does not require continuous fragments on the main chain, but also minimizes the cut weight while adhering to the balance, gap, and charge constraints.
\end{itemize}
 
\section{Related Work}
\label{sec:rel-work}
\subsection{General-purpose graph partitioning}
General-purpose graph partitioning tools only require the adjacency information of the graph
and no additional problem-related information. For special inputs (very small $n$ or $k=2$ and small cuts) sophisticated methods from mathematical 
programming~\cite{DBLP:journals/anor/GhaddarAL11} or using branch-and-bound~\cite{DBLP:journals/mp/DellingFGRW15}
are feasible -- and give provably optimal results. To be of general practical use, in particular for larger instances, most widely used tools
employ local heuristics within a multilevel approach, though (see the survey by
Buluc \etal~\cite{DBLP:journals/corr/BulucMSSS13}). 

The multilevel metaheuristic, popularized for graph partitioning in the mid-1990s~\cite{HendricksonLeland95multilevel}, is a powerful technique and consists of three phases:
First, one computes a hierarchy of graphs $G_0, \dots, G_l$ by recursive coarsening in the first phase. 
$G_l$ ought to be small in size, but topologically similar to the input graph $G_0$. A very good initial 
solution for $G_l$ is computed in the second phase. After that, the recursive coarsening is undone and
the solution prolongated to the next-finer level. In this final phase, in successive steps, the respective prolongated solution on each level is improved using 
local search.

A popular local search algorithm for the third phase of the multilevel process is based on the 
method by Fiduccia and Mattheyses (FM)~\cite{FM82} (many others exist, see~\cite{DBLP:journals/corr/BulucMSSS13}). The main idea 
of FM is to exchange nodes between 
blocks in the order of the cost reductions possible, while maintaining a balanced partition. 
After every node has been moved once, 
the solution with the best cost improvement is chosen. Such a phase is repeated several times, each running in time $\bigO(m)$.

\subsection{Methods for subsystem quantum chemistry}
While this work is based on the \emph{molecular fractionation with conjugate cap} (MFCC) scheme~\cite{mfcc-1,he_fragment_2014}, several more sophisticated approaches have been developed which allow to decrease the size of the error in subsystem quantum-mechanical calculations~\cite{fedorov_exploring_2012,fmo-review,jacob_subsystem_2014}. 
The general idea is to reintroduce the interactions missed by the fragmentation of the supermolecule. A prominent example is the \emph{frozen density embedding} (FDE) approach\shorten{where the density of the subsystem and its environment is relaxed in alternating \emph{freeze and thaw} cycles}~\cite{tomasz-cpl-1996,3fde-2008,jacob_subsystem_2014}. 
All these methods strongly depend on the underlying fragmentation of the supermolecule and it is therefore desirable to minimize the error in the form of the cut
weight itself. Thus, the implementation shown in this paper is applicable to all quantum-chemical subsystem methods needing molecule fragments as an input.

 \section{Solving Main Chain Partitioning Optimally}
 \label{sec:algo-main-chain}
 \newcommand{\tableVar}{\ensuremath{\mathrm{partCut}}}
 \newcommand{\cutCost}{\ensuremath{c}}
 \newcommand{\pred}{\ensuremath{\mathrm{pred}}}

As discussed in the introduction, a protein consists of a main chain, which is folded to yield its characteristic spatial structure.
 Aligning a partition along the main chain uses the locality information in the node order and minimizes the number of cap molecules necessary for a given number of fragments. 
 The problem description from Section~\ref{sec:problem} -- finding fragments with continuous node IDs -- is equivalent to finding a set of $k-1$ \emph{delimiter nodes} $v_{d_1}, v_{d_2}, ... v_{d_{k-1}}$ that separate the fragments. Note that this is not a vertex separator, instead the delimiter nodes induce a set of cut edges due to the continuous node IDs. More precisely, delimiter node $v_{d_j}$ belongs to fragment $j$, $1 \leq j \leq k-1$. 

Consider the delimiter nodes in ascending order.
Given the node $v_{d_2}$, the optimal placement of node $v_{d_1}$ only depends on edges among nodes $u < v_{d_2}$, since all edges $\{u,v\}$ from nodes $u < v_{d_2}$ to nodes $v > v_{d_2}$ are cut no matter where $v_{d_1}$ is placed.
Placing node $v_{d_2}$ thus induces an optimal placement for $v_{d_1}$, using only information from edges to nodes $u < v_{d_2}$.
With this dependency of the positions of $v_{d_1}$ and $v_{d_2}$, placing node $v_{d_3}$ similarly induces an optimal choice for $v_{d_2}$ and $v_{d_1}$, using only information from nodes smaller than $v_{d_3}$. The same argument can be continued inductively for nodes $v_{d_4} \dots v_{d_k}$.

Algorithm~\ref{algo:dynamic} is our dynamic-programming-based solution to the main chain partitioning problem. It uses the property stated above to iteratively compute the optimal placement of $v_{d_{j-1}}$ for all possible values of $v_{d_{j}}$. Finding the optimal placements of $v_{d_1}, \dots v_{d_{j-1}}$ given a delimiter $v_{d_{j}}$ at node $i$ is equivalent to the subproblem of partitioning the first $i$ nodes into $j$ fragments, for increasing values of $i$ and $j$. If $n$ nodes and $k$ fragments are reached, the desired global solution is found.
 We allocate (Line~\ref{line:dynamic:allocate-table}) and fill an $n\times k$ table \tableVar\ with the optimal values for the subproblems.
More precisely, the table entry $\tableVar[i][j]$ denotes the minimum cut weight of a $j$-partition of the first $i$ nodes:
\begin{lemma}
 After the execution of Algorithm~\ref{algo:dynamic}, \tableVar$[i][j]$ contains the minimum cut value for a continuous $j$-partition of the first $i$ nodes. If such a partition is impossible, \tableVar$[i][j]$ contains $\infty$.
\label{lemma:dynamic-correctness}
\end{lemma}

We prove the lemma after describing the algorithm. After the initialization of data structures in Lines~\ref{line:dynamic:allocate-partition} and \ref{line:dynamic:allocate-table},
 the initial values are set in Line~\ref{line:dynamic:initialize-table}: A partition consisting of only one fragment has a cut weight of zero.

 All further partitions are built from a \emph{predecessor partition} and a new fragment. A $j$-partition $\Pi_{i,j}$ of the first $i$ nodes consists of the $j$th fragment and a $(j-1)$-partition with fewer than $i$ nodes.
 A valid predecessor partition of $\Pi_{i,j}$ is a partition $\Pi_{l,j-1}$ of the first $l$ nodes, with $l$ between $i-\maxSize$ and $i-1$.
Node charges have to be taken into account when compiling the set of valid predecessors.
If a backwards search for $\Pi_{i,j}$ from node $i$ encounters two charged nodes $a$ and $b$ with $a<b$, all valid predecessors of $\Pi_{i,j}$ contain at least node $a$ (Line~\ref{line:dynamic:charged-window}).

The additional cut weight induced by adding a fragment containing the nodes $[l+1,i]$ to a predecessor partition $\Pi_{l,j-1}$ is the weight sum of edges connecting nodes in $[1,l]$ to nodes in $[l+1,i]$: $\cutCost{}[l][i] = \sum_{\{u,v\} \in E, u \in [1,l], v \in [l+1,i]} w(u,v)$.
Line~\ref{line:dynamic:cut-cost-table} computes this weight difference for the current node $i$ and all valid predecessors $l$.

For each $i$ and $j$, the partition $\Pi_{i,j}$ with the minimum cut weight is then found in Line~\ref{line:dynamic:cut} by iterating backwards over all valid predecessor partitions and selecting the one leading to the minimum cut. To reconstruct the partition, we store the predecessor in each step (Line~\ref{line:dynamic:pred}). If no partition with the given values is possible, the corresponding entry in \tableVar\ remains at $\infty$.

After the table is filled, the resulting minimum cut weight is at $\tableVar[n][k]$, the corresponding partition is found by following the predecessors (Line~\ref{line:dynamic:follow-predecessors}).

 \begin{algorithm}[tb]
 \KwIn{Graph $G=(V,E)$, fragment count $k$, bool list \emph{isCharged}, imbalance $\epsilon$}
 \KwOut{partition $\Pi$}
 \maxSize = $\lceil |V|/k \rceil \cdot (1+\epsilon)$\;
 allocate empty partition $\Pi$\;\label{line:dynamic:allocate-partition}

 \tableVar[i][j] = $\infty, \forall i \in [1,n], \forall j \in [1,k]$\;\label{line:dynamic:allocate-table}
 \tcc{initialize empty table \tableVar{} with $n$ rows and $k$ columns}
 \tableVar[i][1] = 0, $\forall i \in [1, \maxSize]$\;\label{line:dynamic:initialize-table}

\For{$1 \leq i \leq n$}{\label{line:dynamic:outer-loop}
  windowStart = $\max(i-\maxSize, 1)$\;\label{line:dynamic:windowStart}
  if necessary, increase windowStart so that [windowStart, i] contains at most one charged node\;\label{line:dynamic:charged-window}
  compute column $i$ of cut cost table $\cutCost{}$\;\label{line:dynamic:cut-cost-table}
  \For{$2 \leq j \leq k $}{\label{line:dynamic:inner-loop}
    \tableVar[i][j] = $\min_{l \in [windowStart, i]} \tableVar[l][j-1] + \cutCost{}[l][i]$\;\label{line:dynamic:cut}
    \pred[i][j] = $\argmin_{l \in [windowStart, i]} \tableVar[l][j-1] + \cutCost{}[l][i]$\;\label{line:dynamic:pred}
  }
}

$i = n$\;
\For{$j = k$; $j \geq 2$; $j-=1$}{
  $\mathit{nextI} = \pred[i][j]$\;\label{line:dynamic:follow-predecessors}
  assign nodes between \emph{nextI} and $i$ to fragment $\Pi_j$\;\label{line:dynamic:assign-partition}
  i = \emph{nextI}\;
}
 \Return{$\Pi$}
\caption{Main Chain Partitioning with Dynamic Programming}
\label{algo:dynamic}
 \end{algorithm}

We are now ready to prove Lemma~\ref{lemma:dynamic-correctness} and the algorithm's correctness and time complexity.
\begin{proof}[of Lemma~\ref{lemma:dynamic-correctness}]
By induction over the number of partitions $j$. \\[0.75ex]
\emph{Base Case: $j = 1, \forall i$.}
A 1-partition is a continuous block of nodes. The cut value is zero exactly if the first $i$ nodes contain at most one charge and $i$ is not larger than $\maxSize$. This cut value is written into $\tableVar$\ in Lines~\ref{line:dynamic:allocate-table} and \ref{line:dynamic:initialize-table} and not changed afterwards. \\[0.75ex]
\emph{Inductive Step: $j-1 \rightarrow j$.}
Let $i$ be the current node: A cut-minimal $j$-partition $\Pi_{i,j}$ for the first $i$ nodes contains a cut-minimal $(j-1)$-partition $\Pi_{i',j-1}$ with continuous node blocks.
If $\Pi_{i',j-1}$ were not minimum, we could find a better partition $\Pi_{i',j-1}'$ and use it to improve $\Pi_{i,j}$, a contradiction to $\Pi_{i,j}$ being cut-minimal.
Due to the induction hypothesis, $\tableVar[l][j-1]$ contains the minimum cut value for all node indices $l$, which includes $i'$.
The loop in Line~\ref{line:dynamic:cut} iterates over possible predecessor partitions $\Pi_{l,j-1}$ and selects the one leading to the minimum cut after node $i$.
Given that partitions for $j-1$ are cut-minimal, the partition whose weight is stored in $\tableVar[i][j]$ is cut-minimal as well.

If no allowed predecessor partition with a finite weight exists, $\tableVar[i][j]$ remains at infinity.\hfill\qed 
\end{proof}

\label{subsec:algo-main-chain-complexity}
\begin{theorem}
Algorithm~\ref{algo:dynamic} computes the optimal main chain partition in time $\bigO(n^2 \cdot \maxSize)$.
\label{thm:dynamic-optimality}
\end{theorem}
\begin{proof}
The correctness in terms of optimality follows directly from Lemma~\ref{lemma:dynamic-correctness}. We thus continue with establishing the time complexity.
The nested loops in Lines~\ref{line:dynamic:outer-loop} and~\ref{line:dynamic:inner-loop} require $\bigO(n\cdot k)$ iterations in total.
Line~\ref{line:dynamic:charged-window} is executed $n$ times and has a complexity of $\maxSize$.
At Line~\ref{line:dynamic:cut} in the inner loop, up to \maxSize{} predecessor partitions need to be evaluated, each with two constant time table accesses.
Computing the cut weight column $\cutCost{}[\cdot][i]$ for fragments ending at node $i$ (Line~\ref{line:dynamic:cut-cost-table}) involves summing over the edges of $\bigO(\maxSize)$ predecessors, each having at most $\bigO(n)$ neighbors. Since the cut weights constitute a reverse prefix sum, the column $\cutCost{}[\cdot][i]$ can be computed in $\bigO(n\cdot\maxSize)$ time by iterating backwards.
Line~\ref{line:dynamic:cut-cost-table} is executed $n$ times, leading to a total complexity of $\bigO(n^2\cdot\maxSize)$.
Following the predecessors and assigning nodes to fragments is possible in linear time, thus the $\bigO(n^2\cdot\maxSize)$ to compile the cut cost table dominates the running time.\hfill\qed
\end{proof}

 \section{Algorithms for General Protein Partitioning}
 \label{sec:algo-general}
 As discussed in Section~\ref{sec:problem}, one cannot use general-purpose graph partitioning programs due to the new constraints required by the DFT calculations. Moreover, if the constraint of the previous section is dropped, the DP-based algorithm is not optimal in general any more.
Thus, we propose three algorithms for the general problem in this section: The first two, a greedy agglomerative method and Multilevel-FM, build on existing graph
partitioning knowledge but incorporate the new constraints directly into the optimization process. 
The third one is a simple postprocessing repair procedure that works in many cases. It takes the output of a traditional graph partitioner and fixes it
so as to fulfill the constraints.

 \subsection{Greedy Agglomerative Algorithm}
 \label{subsec:greedy}
 The greedy agglomerative approach, shown in Algorithm~\ref{algo:greedy} in Appendix~\ref{sec:add-pseudo}, is similar in spirit to Kruskal's MST algorithm and 
to approaches proposed for clustering graphs with respect to the objective function modularity~\cite{clauset2004finding}. 
It initially sorts edges by weight and puts each node into a singleton fragment.
 Edges are then considered iteratively with the heaviest first; the fragments belonging to the incident nodes are merged if no constraints are violated. This is repeated until no edges are left or the desired fragment count is achieved.

The initial edge sorting takes $\bigO(m\log m)$ time. Initializing the data structures is possible in linear time. The main loop (Line~\ref{line:greedy-main-loop}) has at most $m$ iterations. Checking the size and charge constraints is possible in constant time by keeping arrays of fragment sizes and charge states. The time needed for checking the gaps and merging is linear in the fragment size and thus at most $\bigO(\maxSize)$.

 The total time complexity of the greedy algorithm is thus:
  \[
 T(\mathrm{Greedy}) \in \bigO(m \cdot \max{\{\mathrm{\maxSize}, \log m\}}).
 \]

\subsection{Multilevel Algorithm with Fiduccia-Mattheyses Local Search}
\label{subsec:multilevel}

Algorithm~\ref{algo:multilevel} (Appendix \ref{sec:add-pseudo}) is similar to existing multilevel partitioners using non-binary (\ie $k>2$) Fiduccia-Mattheyses
(FM) local search. Our adaptation incorporates the constraints throughout the whole partitioning process, though. 
First a hierarchy of graphs $G_0, G_1, \dots G_l$ is created by recursive coarsening (Line~\ref{line:multilevel:coarsened-hierarchy}).
The edges contracted during coarsening are chosen with a local matching strategy.
An edge connecting two charged nodes stays uncontracted, thus ensuring that a fragment contains at most one charged node even in the coarsest partitioning phase. 
The coarsest graph is then partitioned into $\Pi_l$ using region growing or recursive bisection.
If an optional input partition $\Pi '$ is given, it is used as a guideline during coarsening and replaces $\Pi_l$ if it yields a better cut.
We execute both our greedy and DP algorithm and use the partition with the better cut as input partition $\Pi '$ for the multilevel algorithm.

After obtaining a partition for the coarsest graph, the graph is iteratively uncoarsened and the partition projected to the next finer level.
We add a rebalancing step at each level (Line~\ref{line:multilevel:rebalance}), since a non-binary FM step does not guarantee balanced partitions if the input is imbalanced.
A Fiduccia-Mattheyses step is then performed to yield local improvements (Line~\ref{line:multilevel:fm}):
For a partition with $k$ fragments, this non-binary FM step consists of one priority queue for each fragment.
Each node $v$ is inserted into the priority queue of its current fragment, the maximum gain (\ie reduction in cut weight when $v$ is moved to another fragment) is used as key.
While at least one queue is non-empty, the highest vertex of the largest queue is moved if the constraints are still fulfilled, and the movement recorded.
After all nodes have been moved, the partition yielding the minimum cut is taken.
In our variant, nodes are only moved if the charge constraint stays fulfilled.
 
 \subsection{Repair Procedure}
 \label{subsec:repair}
As already mentioned, traditional graph partitioners produce in general solutions that do not adhere to the constraints for protein partitioning.
To be able to use existing tools, however, we propose a simple repair procedure for an existing partition which possibly does not fulfill the 
charge, gap, or balance constraints. To this end, Algorithm~\ref{algo:repair} in Appendix~\ref{sec:add-pseudo} performs one sweep over all nodes (Line~\ref{line:repair-outer-loop}) and
checks for every node $v$ whether the constraints are violated at this point.
If they are and $v$ has to be moved, an FM step is performed: Among all fragments that could possibly receive $v$, the one minimizing the cut weight is selected. If no suitable target fragment exists, a new singleton fragment is created. Note that due to the local search, this step can lead to more than $k$ fragments, even if a partition with $k$ fragments is possible.

The cut weight table allocated in Line~\ref{line:repair:cutWeight-definition} takes $\bigO(n\cdot k + m)$ time to create.
Whether a constraint is violated can be checked in constant time per node by counting the number of nodes and charges observed for each fragment.
A node needs to be moved when at least one charge or at least $\maxSize$ nodes have already been encountered in the same fragment.
Finding the best target partition (Line~\ref{line:repair:find-target}) takes $\bigO(k)$ iterations, updating the cut weight table after moving a node $v$ is linear in the degree $\deg(v)$ of $v$.
The total time complexity of a repair step is thus: $\bigO(n\cdot k + m + n \cdot k + \sum_{v} \deg(v))$ = $\bigO(n\cdot k + m)$.

 \section{Experiments}
 \label{sec:experiments}
 \subsection{Settings}
 \label{subsec:settings}
 We evaluate our algorithms on graphs derived from several proteins and compare the resulting cut weight.
 As main chain partitioning is a special case of general protein partitioning, the solutions generated by our dynamic programming algorithm are valid solutions of the general problem, though perhaps not optimal. Other algorithms evaluated are Algorithm~\ref{algo:greedy} (Greedy), \ref{algo:multilevel} (Multilevel), and the external partitioner KaHiP~\cite{sandersschulz2013}, used with the repair step discussed in Section~\ref{subsec:repair}. 
 The algorithms are implemented in C++ and Python using the NetworKit tool suite~\cite{DBLP:journals/corr/StaudtSM14}, the source code is available from a hg repository\footnote{\url{https://algohub.iti.kit.edu/parco/NetworKit/NetworKit-chemfork/}}.

 We use graphs derived from five common proteins, covering the most frequent structural properties. Ubiquitin~\cite{ubi_pdb} (also see Figure~\ref{fig:3d_ubiquitin} in Appendix~\ref{sec:illustrate}) and the Bubble Protein~\cite{bubble_pdb} are rather small proteins with 76 and 64 amino acids, respectively. Due to their biological functions, their overall size and their diversity in the contained structural features, they are commonly used as test cases for quantum-chemical subsystem methods~\cite{karin-3fde}.
The Green Fluorescent Protein (GFP)~\cite{gfp_pdb} plays a crucial role in the bioluminescence of marine organisms and is widely expressed in other organisms as a fluorescent label for microscopic techniques.
Like the latter one, Bacteriorhodopsin (bR)~\cite{br_pdb} and the Fenna-Matthews-Olson protein (FMO)~\cite{fmo_pdb} are large enough to render quantum-chemical calculations on the whole proteins practically infeasible. Yet, investigating them with quantum-chemical methods is key to understanding the photochemical processes they are involved in.
The graphs derived from the latter three proteins have 225, 226 and 357 nodes, respectively. They are complete graphs with weighted $n(n-1)/2$ edges.
All instances can be found in the mentioned hg repository in folder \texttt{input/}.

In our experiments we partition the graphs into fragments of different sizes (\ie we vary the fragment number $k$). The small proteins ubiquitin and bubble are partitioned into 2, 4, 6 and 8 fragments, leading to fragments of average size 8-38. The other proteins are partitioned into 8, 12, 16, 20 and 24 fragments, yielding average sizes between 10 and 45.
As maximum imbalance, we use values for $\epsilon$ of 0.1 and 0.2. While this may be larger than usual values of $\epsilon$ in graph partitioning, fragment sizes in our case are comparably small and an imbalance of 0.1 is possibly reached with the movement of a single node.

On these proteins, the running time of all partitioning implementations is on the order of a few seconds on a commodity laptop, we therefore omit detailed time measurements. 

 \paragraph{Charged Nodes.}
 Depending on the environment, some of the amino acids are charged. As discussed in Section~\ref{sec:problem}, at most one charge is allowed per fragment.
 We repeatedly sample $\lfloor 0.8\cdot k \rfloor$ random charged nodes among the potentially charged, under the constraint that a valid main chain partition is still possible.
 To smooth out random effects, we perform 20 runs with different random nodes charged. Introducing charged nodes may cause the naive partition to become invalid. In these cases, we use the repair procedure on the invalid naive partition and compare the cut weights of other algorithms with the cut weight of the repaired naive partition.
 \subsection{Results}
\begin{figure}
\includegraphics{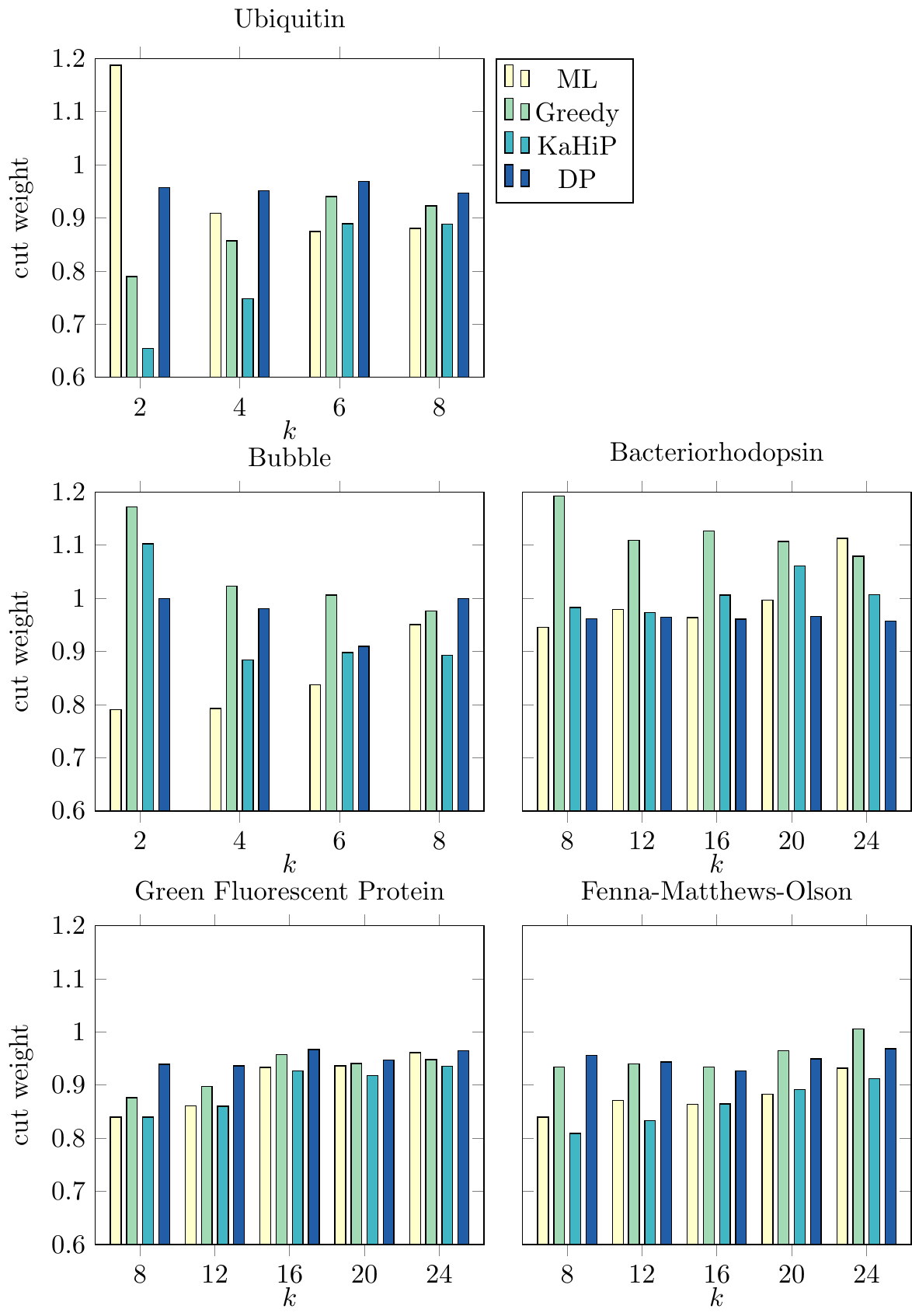}
\caption{Comparison of partitions given by several algorithms and proteins, for $\epsilon=0.1$. The partition quality is measured by the cut weight in comparison to the naive solution.}
\label{plot:uncharged-comparison}
\end{figure}

For the uncharged scenario, Figure~\ref{plot:uncharged-comparison} shows a comparison of cut weights for different numbers of fragments and a maximum imbalance of 0.1. The cut weight is up to 34.5\% smaller than with the naive approach (or 42.8\% with $\epsilon = 0.2$, see Figure~\ref{plot:uncharged-comparison-epsilon-0.2}). The best algorithm choice depends on the protein: For ubiquitin, green fluorescent protein, and Fenna-Matthew-Olson protein, the external partitioner KaHiP in combination with the repair step described in Section~\ref{subsec:repair} gives the lowest cut weight when averaged over different fragment sizes. For the bubble protein, the multilevel algorithm from Section~\ref{subsec:multilevel} gives on average the best result, while for bacteriorhodopsin, the best cut weight is achieved by the dynamic programming (DP) algorithm.
The DP algorithm is always as least as good as the naive approach.
This already follows from Theorem~\ref{thm:dynamic-optimality}, as the naive partition is aligned along the main chain and thus found by DP in case it is optimal.
DP is the only algorithm with this property, all others perform worse than the naive approach for at least one combination of parameters.

The general intuition that smaller fragment sizes leave less room for improvements compared to the naive solution is confirmed by our experimental results.
Figure~\ref{plot:uncharged-comparison-epsilon-0.2} (Appendix ~\ref{sec:illustrate}) shows the comparison with imbalance $\epsilon = 0.2$.
While the general trend is similar and the best choice of algorithm depends on the protein, the cut weight is usually more clearly improved.
Moreover, a meta algorithm that executes all single algorithms and picks their best solution yields average improvements (geometric mean)
of $13.5\%, 16\%,$ and $20\%$ for $\epsilon = 0.1, 0.2$, and $0.3$, respectively, compared to the naive reference.
Such a meta algorithm requires only about ten seconds per instance, negligible in the whole DFT workflow.

Randomly charging nodes changes the results only insignificantly, as seen in Figure~\ref{plot:charged-epsilon-0.1}.
The necessary increase in cut weight for the algorithm's solutions is likely compensated by a similar increase in the naive partition due to the necessary repairs.
 
 \section{Conclusions}
 \label{sec:conclusions}
 Partitioning protein graphs for subsystem quantum-chemistry is a new problem with unique constraints which general-purpose graph partitioning algorithms were unable to handle. We have provided several algorithms for this problem and proved the optimality of one in the special case of partitioning along the main chain. 
With our algorithms chemists are now able to address larger problems in an automated manner with smaller error. Larger proteins, in turn, in connection with a reasonable imbalance, may provide more opportunities for improving the quality of the naive solution further.

 \clearpage
\bibliographystyle{plain}
\bibliography{paper,Bibliography,zotero}

\begin{thebibliography}{10}

\bibitem{Andreev2006}
Konstantin Andreev and Harald Racke.
\newblock Balanced graph partitioning.
\newblock {\em Theory of Computing Systems}, 39(6):929--939, 2006.

\bibitem{DBLP:journals/corr/BulucMSSS13}
Aydin Bulu{\c{c}}, Henning Meyerhenke, Ilya Safro, Peter Sanders, and Christian
  Schulz.
\newblock Recent advances in graph partitioning.
\newblock {\em Accepted as Chapter in Algorithm Engineering, Overview Paper
  concerning the DFG SPP 1307}, 2016.
\newblock Preprint available at http://arxiv.org/abs/1311.3144.

\bibitem{clauset2004finding}
A.~Clauset, M.E.J. Newman, and C.~Moore.
\newblock Finding community structure in very large networks.
\newblock {\em Physical Review E}, 70(6):66111, 2004.

\bibitem{cramer-book}
Christopher~J. Cramer.
\newblock {\em {E}ssentials of {C}omputational {C}hemistry}.
\newblock Wiley, New York, 2002.

\bibitem{DBLP:journals/mp/DellingFGRW15}
Daniel Delling, Daniel Fleischman, Andrew~V. Goldberg, Ilya Razenshteyn, and
  Renato~F. Werneck.
\newblock An exact combinatorial algorithm for minimum graph bisection.
\newblock {\em Math. Program.}, 153(2):417--458, 2015.

\bibitem{fmo-review}
Dmitri~G. Fedorov and Kazuo Kitaura.
\newblock {E}xtending the {P}ower of {Q}uantum {C}hemistry to {L}arge {S}ystems
  with the {F}ragment {M}olecular {O}rbital {M}ethod.
\newblock {\em J. Phys. Chem. A}, 111:6904--6914, 2007.

\bibitem{fedorov_exploring_2012}
Dmitri~G. Fedorov, Takeshi Nagata, and Kazuo Kitaura.
\newblock {E}xploring chemistry with the fragment molecular orbital method.
\newblock {\em Phys. Chem. Chem. Phys.}, 14:7562--7577, 2012.

\bibitem{FM82}
C.~Fiduccia and R.~Mattheyses.
\newblock A linear time heuristic for improving network partitions.
\newblock In {\em Proc.\ 19th {ACM}/{IEEE} Design Automation Conf.}, pages
  175--181, Las Vegas, NV, June 1982.

\bibitem{adf-lin-scaling}
C.~Fonseca~Guerra, J.~G. Snijders, G.~te~Velde, and E.~J. Baerends.
\newblock {T}owards an order-{N} {D}{F}{T} method.
\newblock {\em Theor. Chem. Acc.}, 99:391, 1998.

\bibitem{GareyJS74some}
M.~R. Garey, D.~S. Johnson, and L.~Stockmeyer.
\newblock Some simplified {NP}-complete problems.
\newblock In {\em Proceedings of the 6th Annual ACM Symposium on Theory of
  Computing (STOC'74)}, pages 47--63. ACM Press, 1974.

\bibitem{DBLP:journals/anor/GhaddarAL11}
Bissan Ghaddar, Miguel~F. Anjos, and Frauke Liers.
\newblock A branch-and-cut algorithm based on semidefinite programming for the
  minimum \emph{k}-partition problem.
\newblock {\em Annals {OR}}, 188(1):155--174, 2011.

\bibitem{gordon_fragmentation_2012}
Mark~S. Gordon, Dmitri~G. Fedorov, Spencer~R. Pruitt, and Lyudmila~V.
  Slipchenko.
\newblock {F}ragmentation {M}ethods: {A} {R}oute to {A}ccurate {C}alculations
  on {L}arge {S}ystems.
\newblock {\em Chem. Rev.}, 112:632--672, 2012.

\bibitem{he_fragment_2014}
Xiao He, Tong Zhu, Xianwei Wang, Jinfeng Liu, and John Z.~H. Zhang.
\newblock {F}ragment {Q}uantum {M}echanical {C}alculation of {P}roteins and
  {I}ts {A}pplications.
\newblock {\em Acc. Chem. Res.}, 47:2748--2757, 2014.

\bibitem{HendricksonLeland95multilevel}
B.~Hendrickson and R.~Leland.
\newblock A multi-level algorithm for partitioning graphs.
\newblock In {\em Proceedings Supercomputing '95}, page 28 (CD). ACM Press,
  1995.

\bibitem{pyadf-2011}
{\relax Ch}ristoph~R Jacob, S.~Maya Beyhan, Rosa~E Bulo, Andr\'{e}
  Severo~Pereira Gomes, Andreas~W G\"{o}tz, Karin Kiewisch, Jetze Sikkema, and
  Lucas Visscher.
\newblock {P}y{A}{D}{F} {\textemdash} {A} scripting framework for multiscale
  quantum chemistry.
\newblock {\em J. Comput. Chem.}, 32:2328--2338, 2011.

\bibitem{jacob_subsystem_2014}
{\relax Ch}ristoph~R. Jacob and Johannes Neugebauer.
\newblock {S}ubsystem density-functional theory.
\newblock {\em WIREs Comput. Mol. Sci.}, 4:325--362, 2014.

\bibitem{3fde-2008}
{\relax Ch}ristoph~R. Jacob and Lucas Visscher.
\newblock {A} subsystem density-functional theory approach for the quantum
  chemical treatment of proteins.
\newblock {\em J. Chem. Phys.}, 128:155102, 2008.

\bibitem{jensen-book}
Frank Jensen.
\newblock {\em {I}ntroduction to {C}omputational {C}hemistry}.
\newblock Wiley \& Sons, Chichester, 2nd edition, 2007.

\bibitem{karin-3fde}
Karin Kiewisch, {\relax Ch}ristoph~R. Jacob, and Lucas Visscher.
\newblock {Q}uantum-{C}hemical {E}lectron {D}ensities of {P}roteins and of
  {S}elected {P}rotein {S}ites from {S}ubsystem {D}ensity {F}unctional
  {T}heory.
\newblock {\em J. Chem. Theory Comput.}, 9:2425--2440, 2013.

\bibitem{br_pdb}
Janos~K. Lanyi and Brigitte Schobert.
\newblock Structural changes in the l photointermediate of bacteriorhodopsin.
\newblock {\em Journal of Molecular Biology}, 365(5):1379 -- 1392, 2007.

\bibitem{ochsenfeld_linear-scaling_2007}
Christian Ochsenfeld, Jorg Kussmann, and D.~S. Lambrecht.
\newblock {L}inear-{S}caling {M}ethods in {Q}uantum {C}hemistry.
\newblock In {\em Reviews in Computational Chemistry}, volume~23, pages 1--82.
  Wiley-VCH, New York, 2007.

\bibitem{bubble_pdb}
Johan~Gotthardt Olsen, Claus Flensburg, Ole Olsen, Gerard Bricogne, and Anette
  Henriksen.
\newblock {Solving the structure of the bubble protein using the anomalous
  sulfur signal from single-crystal in-house Cu{\it K}{$\alpha$} diffraction
  data only}.
\newblock {\em Acta Crystallographica Section D}, 60(2):250--255, 2004.

\bibitem{gfp_pdb}
Mats Orm{\"o}, Andrew~B. Cubitt, Karen Kallio, Larry~A. Gross, Roger~Y. Tsien,
  and S.~James Remington.
\newblock Crystal structure of the aequorea victoria green fluorescent protein.
\newblock {\em Science}, 273(5280):1392--1395, 1996.

\bibitem{Pavlopoulos2011}
Georgios~A. Pavlopoulos, Maria Secrier, Charalampos~N. Moschopoulos,
  Theodoros~G. Soldatos, Sophia Kossida, Jan Aerts, Reinhard Schneider, and
  Pantelis~G. Bagos.
\newblock Using graph theory to analyze biological networks.
\newblock {\em BioData Mining}, 4(1):1--27, 2011.

\bibitem{ubi_pdb}
R.~Ramage, J.~Green, T.~W. Muir, O.~M. Ogunjobi, S.~Love, and K.~Shaw.
\newblock Synthetic, structural and biological studies of the ubiquitin system:
  the total chemical synthesis of ubiquitin.
\newblock {\em Biochemical Journal}, 299(1):151--158, 1994.

\bibitem{sandersschulz2013}
Peter Sanders and Christian Schulz.
\newblock {Think Locally, Act Globally: Highly Balanced Graph Partitioning}.
\newblock In {\em Proceedings of the 12th International Symposium on
  Experimental Algorithms (SEA'13)}, volume 7933 of {\em LNCS}, pages 164--175.
  Springer, 2013.

\bibitem{DBLP:journals/corr/StaudtSM14}
Christian Staudt, Aleksejs Sazonovs, and Henning Meyerhenke.
\newblock Networ{K}it: An interactive tool suite for high-performance network
  analysis.
\newblock {\em CoRR}, abs/1403.3005, 2014.

\bibitem{adf}
{Theoretical Chemistry, Vrije Universiteit Amsterdam}.
\newblock \textsc{{A}df}, {A}msterdam density functional program.
\newblock {URL}: http://www.scm.com.

\bibitem{fmo_pdb}
Dale~E. Tronrud and James~P. Allen.
\newblock Reinterpretation of the electron density at the site of the eighth
  bacteriochlorophyll in the fmo protein from pelodictyon phaeum.
\newblock {\em Photosynthesis Research}, 112(1):71--74, 2012.

\bibitem{tomasz-cpl-1996}
Tomasz~Adam Wesolowski and Jacques Weber.
\newblock {K}ohn-{S}ham equations with constrained electron density: an
  iterative evaluation of the ground-state electron density of interaction
  molecules.
\newblock {\em Chem. Phys. Lett.}, 248:71--76, 1996.

\bibitem{mfcc-1}
Da~W. Zhang and J.~Z.~H. Zhang.
\newblock {M}olecular fractionation with conjugate caps for full quantum
  mechanical calculation of protein{\textendash}molecule interaction energy.
\newblock {\em J. Chem. Phys.}, 119:3599--3605, 2003.

\end{thebibliography}
\clearpage
\appendix

\section*{Appendix}

\section{Illustrations and Additional Experimental Results}
\label{sec:illustrate}
 \begin{figure}[h]
  \begin{tikzpicture}
  \begin{axis}[xlabel=number of amino acids in fragment, ylabel={error [arb. u.]}]
   \plot+[blue] table {plots/frag_err_naive.dat};
   \end{axis}
 \end{tikzpicture}
 \caption{Predicted error for interaction energies with naive fragmentation every $X$ amino acids for the small protein ubiquitin. Unpredictable minima and maxima depending on the location of the uniformly distributed cuts occur along the main chain.}
 \label{plot:frag_err_naive}
 \end{figure}
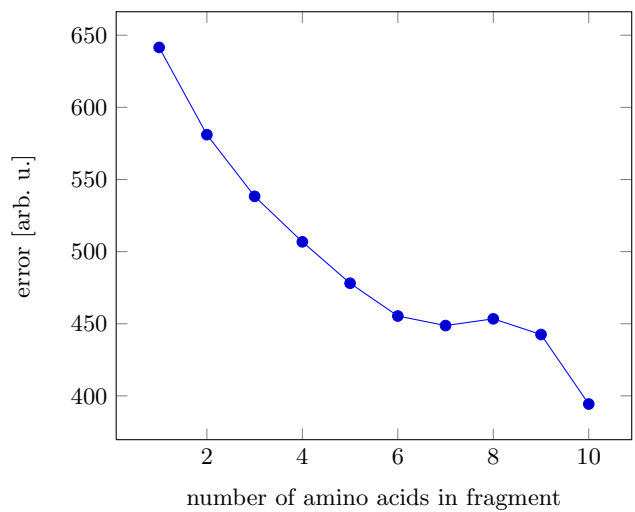

 \begin{figure}[h]
 \includegraphics[width=0.5\textwidth]{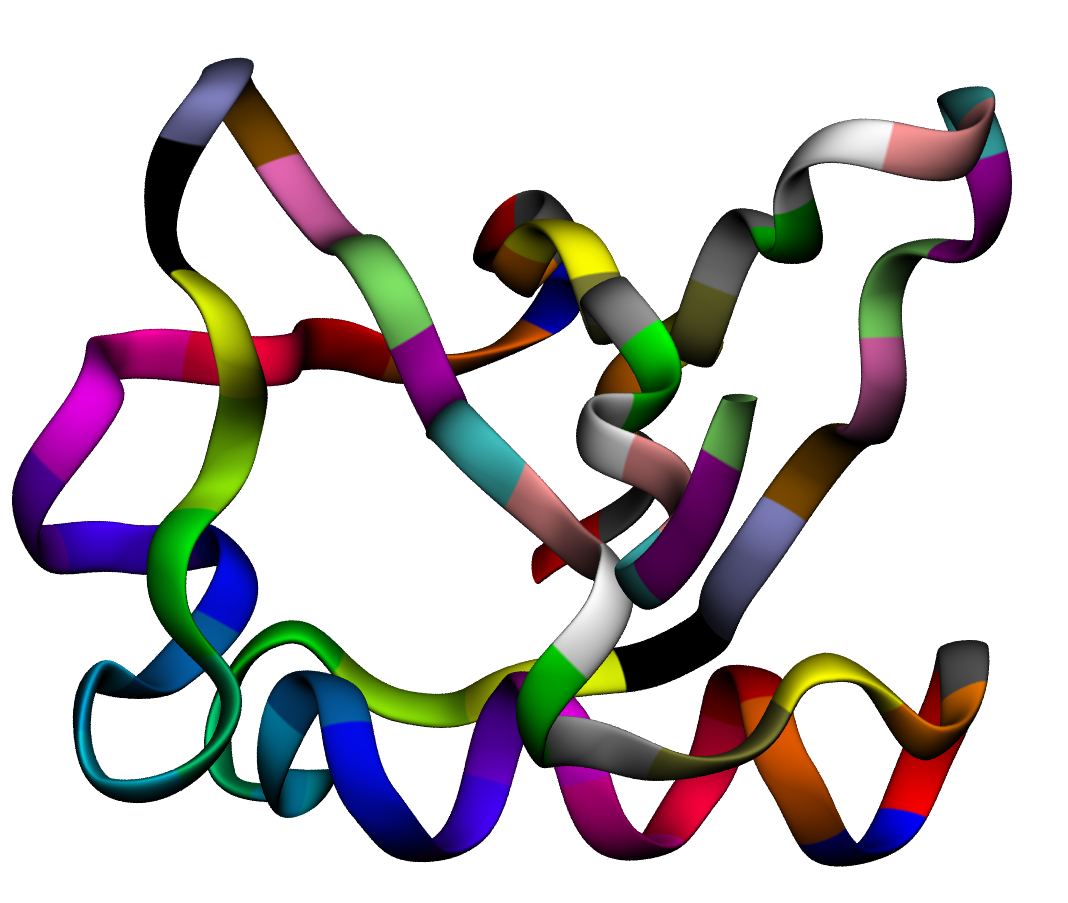}
  \caption{3D-Visualization of Ubiquitin. Single amino acids in different colors. Helical secondary structure at the bottom, beta-sheet like secondary structures in the upper left and right.}
 \label{fig:3d_ubiquitin}
 \end{figure}

\label{appx:run-time-chem}
 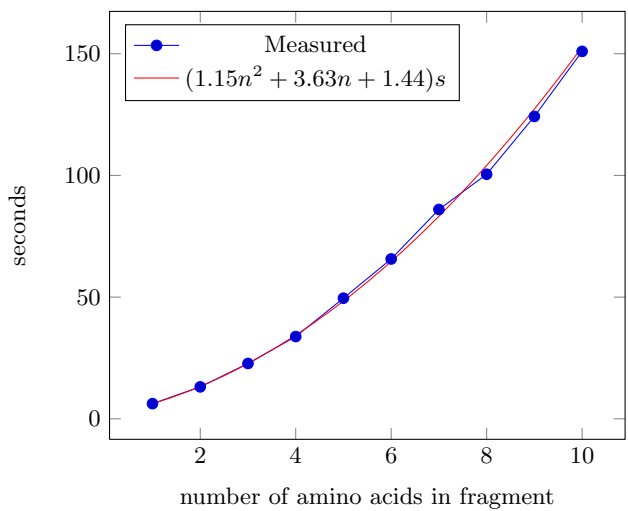
\begin{figure}
  \begin{tikzpicture}
  \begin{axis}[xlabel=number of amino acids in fragment, ylabel=seconds, legend pos=north west]
   \plot+[blue] table {plots/time_frag_mean.dat};\addlegendentry{Measured};
   \plot[red] expression[domain=1:10] {1.15*x^2 + 3.63*x + 1.44};\addlegendentry{$(1.15n^2 + 3.63n + 1.44)s$};
  \end{axis}
 \end{tikzpicture}
 \caption{Time in seconds required for quantum chemical density functional (DFT, BP86, DZP) calculations of protein fragments on 16 Intel Xeon cores (2x Haswell-EP/2640v3/2.6 GHz) executed with pyADF~\cite{pyadf-2011} and ADF program package~\cite{adf}. As seen by the close match of the red fit line, time grows quadratically.}
 \label{plot:time-simulation}
 \end{figure}

\begin{figure}
\includegraphics{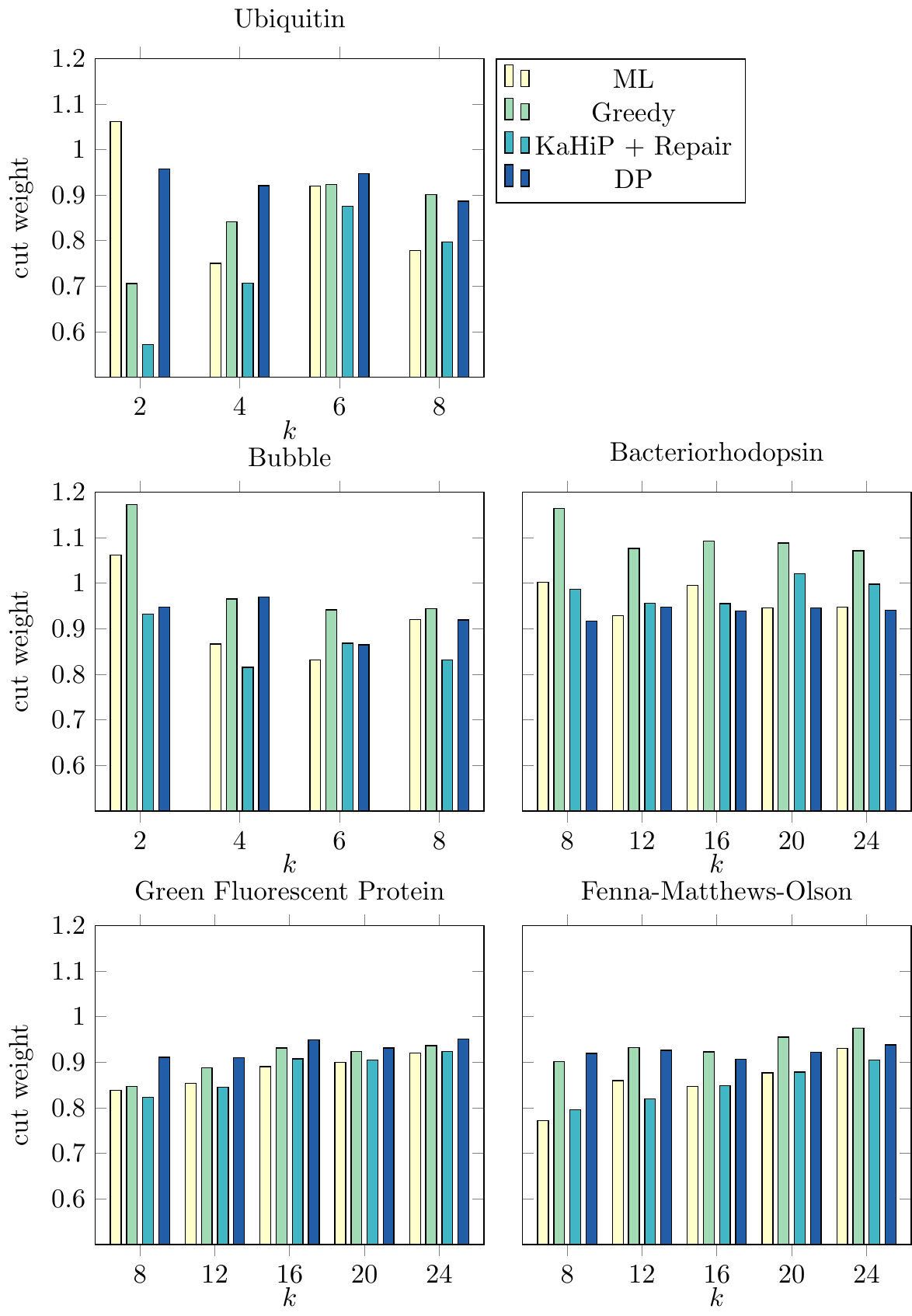}
\caption{Comparison of cut weights for $\epsilon=0.2$.}
\label{plot:uncharged-comparison-epsilon-0.2}
\end{figure}

\begin{figure}
\includegraphics{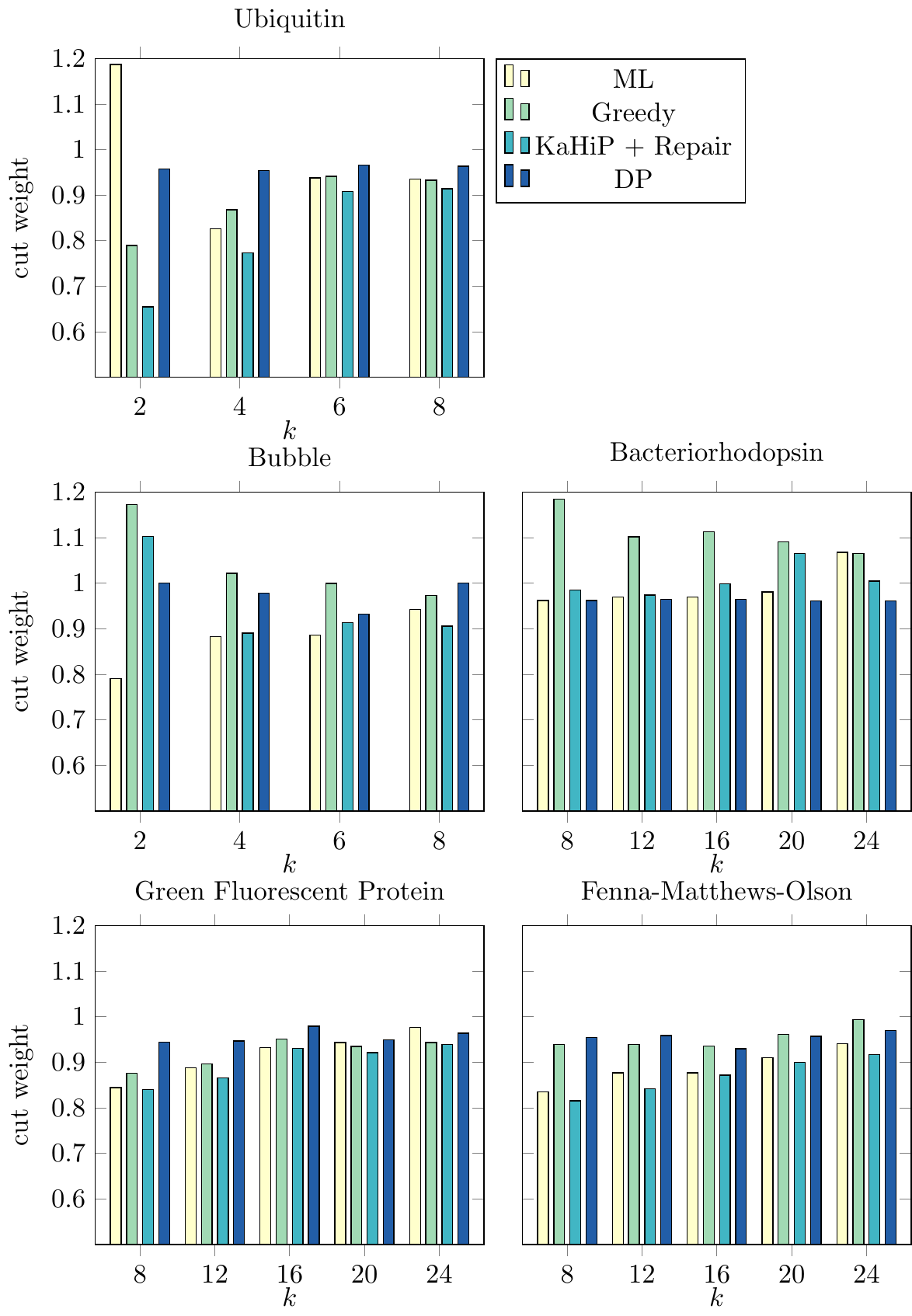}
\caption{Comparison of cut weights for $\epsilon=0.1$ and node charges.}
\label{plot:charged-epsilon-0.1}
\end{figure}

\clearpage
\section{Additional Pseudocodes}
\label{sec:add-pseudo}
 \begin{algorithm}
 \KwIn{Graph $G=(V,E)$, fragment count $k$, list \emph{charged}, imbalance $\epsilon$}
 \KwOut{partition $\Pi$}
  sort edges by weight, descending\;
  $\Pi$ = create one singleton partition for each node\; 
  chargedPartitions = partitions containing a charged node\;
  \maxSize = $\lceil |V|/k \rceil \cdot (1+\epsilon)$\;

  \For{edge $\{u,v\}$}{\label{line:greedy-main-loop}
    allowed = True\;
    \If{$\Pi[u] \in \mathrm{chargedPartitions}$ and $\Pi[v] \in \mathrm{chargedPartitions}$}{
      allowed = False\;
    }
    \If{$|\Pi[u]| + |\Pi[v]| >$ \maxSize}{
      allowed = False\;
    }
    \For{node $x \in \Pi[u] \cup \Pi[v]$}{
      \If{$x+2 \in \Pi[u] \cup \Pi[v]$ and $x+1 \not\in \Pi[u] \cup \Pi[v]$}{
      allowed = False\;
      }
    }
    \If{allowed}{
      merge $\Pi[u]$ and $\Pi[v]$\;
      update chargedPartitions\;
    }
    \If{number of fragments in $\Pi$ equals $k$}{break\;}
  }
  \Return{$\Pi$}
\caption{Greedy Agglomerative Algorithm}
\label{algo:greedy}
 \end{algorithm}

  \begin{algorithm}[tb]
 \KwIn{Graph $G=(V,E)$, fragment count $k$, list \emph{charged}, imbalance $\epsilon$, [$\Pi '$]}
\KwOut{partition $\Pi$}
$G_0, \dots, G_l$ = hierarchy of coarsened Graphs, $G_0 = G$\;\label{line:multilevel:coarsened-hierarchy}
$\Pi_l$ = partition $G_l$ with region growing or recursive bisection\;
\For{$0 \leq i < l$}{
  uncoarsen $G_i$ from $G_{i+1}$\;\label{line:multilevel:uncoarsen}
  $\Pi_i$ = projected partition from $\Pi_{i+1}$\;\label{line:multilevel:project}
  rebalance $\Pi_i$, possibly worsen cut weight\;\label{line:multilevel:rebalance}
  \tcc{Local improvements}
  gain = NaN\;
  \Repeat{gain == 0}{
    oldcut = $\mathrm{cut}(\Pi_i', G)$\;
    $\Pi_i'$ = Fiduccia-Mattheyses-Step of $\Pi_i$ with constraints\;\label{line:multilevel:fm}
    gain = $\mathrm{cut}(\Pi_i', G)$ - oldcut\;
  }
}
\caption{Multilevel-FM}
\label{algo:multilevel}
 \end{algorithm}

  \begin{algorithm}
\KwIn{Graph $G=(V,E)$, $k$-partition $\Pi$, list \emph{charged}, imbalance $\epsilon$}
\KwOut{partition $\Pi'$}

cutWeight[$i][j] = 0, 1 \leq i \leq n, 1 \leq j \leq k$\;\label{line:repair:cutWeight-definition}
\For{edge $\{u,v\}$ in $E$}{
  cutWeight$[u][\Pi(u)] += w(u,v)$\;
  cutWeight$[v][\Pi(v)] += w(u,v)$\;
}

\For{node $v$ in $V$}{\label{line:repair-outer-loop}
  \tcc{Check whether node can stay}
   \If{charge violated \emph{or} size violated \emph{or} gap of size 1}{
     $\Psi$ = set of allowed target fragments\;
     \If{$\Psi$ is empty}{
       create new fragment for $v$\;
     }
     \Else{
     \tcc{Fiduccia-Mattheyses-step: To minimize the cut weight, move the node to the fragment to which it has the strongest connection}
       target = $\argmax_{i\in\Psi} \{\mathrm{cutWeight}[v][i]\}$\;\label{line:repair:find-target}
       move $v$ to target\;
     }
     update charge counter, size counter and cutWeight\;\label{line:repair:update}
   }
}
\caption{Repairing a partition}
\label{algo:repair}
\end{algorithm}
 \end{document}